\documentclass[psamsfonts]{amsart}

\usepackage{amssymb}
\usepackage{amsmath}
\usepackage{amsthm}
\usepackage{mathrsfs}
\usepackage{graphicx}
\usepackage{bbm}
\usepackage{psfrag}
\usepackage{hyperref}

\theoremstyle{plain}
\newtheorem{thm}{Theorem}[section]
\newtheorem{lemma}[thm]{Lemma}
\theoremstyle{definition}
\newtheorem{defi}[thm]{Definition}

\numberwithin{equation}{section}

\begin{document}

\title{Escaping the Brownian stalkers}
\author{Alexander Wei\ss{}}
\address{Weierstra\ss{} Institut f\"{u}r Angewandte Analysis und Stochastik\\
		 Mohrenstra\ss{}e 39\\
		 10117 Berlin\\
		 Germany
}
\email{weiss1@wias-berlin.de}
\thanks{Work supported by the DFG Research Center MATHEON}
\date{\today}
\begin{abstract}
We propose a simple model for the behaviour of longterm investors on a stock market, consisting of three particles, which represent the current price of the stock and the opinion of the buyers, respectively sellers, about the right trading price. As time evolves, both groups of traders update their opinions with respect to the current price. The update speed is controled by a parameter $\gamma$, the price process is described by a geometric Brownian motion. We consider the stability of the market in terms of the distance between the buyers' and sellers' opinion, and prove that the distance process is recurrent/transient in dependence on $\gamma$.
\end{abstract}

\maketitle

\section{Introduction}
In this article we suggest a simple model for the behaviour of longterm investors on a share market. We observe the evolution of three particles. One of them represents the current price of the share, the second one the opinion of shareholders about the share's value, and the last one the opinion of potential buyers. As longterm investors do not speculate on fast returns, it is reasonable to assume two features: first, the value of the share in the eyes of their holders is much higher than the current price, and it is much lower in the eyes of potential buyers. However, both groups of investors will not wait forever. They will modify their opinions in dependence of the price development. But, as second feature, the traders will only slowly adapt to price changes. As opposed to short-time traders, who gamble on returns on short time intervals, there is no need for longterm investors to react on small fluctuations.

Eventually, as the price changes and the investors adjust their opinions, the price will reach the value, which is expected by the traders. We assume a symmetric behaviour of buyers and sellers, and need to consider, what happens if the price reaches the right value in shareholder's opinion. Because the price has reached a {\it fair} level, the investors will sell their shares. At the very moment there are new holders, namely the buyers of the shares. Eventually, the price will drop, and there is again a group of individuals not willing to follow this decrement. This means, while the individuals in the group of longterm investors will change in time, the group itself will persist. Figure \ref{fig:contSystem} shows an example for the evolution of the system on a logarithmic scale. The price is denoted by $B$, the opinion of buyers by $X$, and the one of holders by $Y$.

\begin{figure}[htb]
\begin{center}
\includegraphics[width=.75\textwidth]{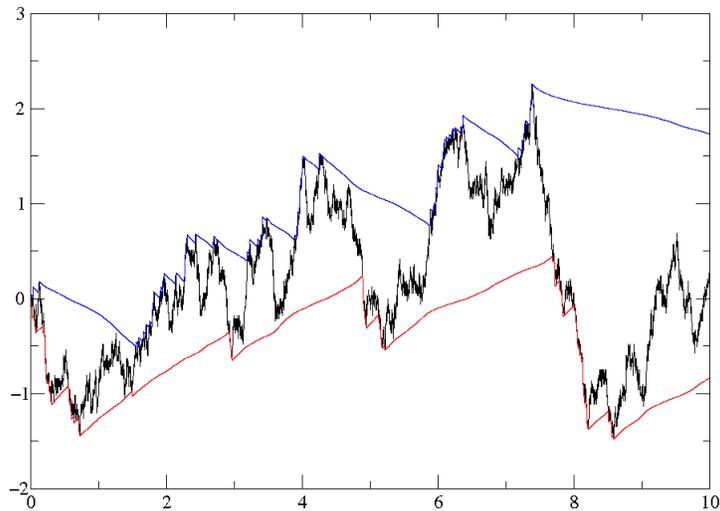}
\caption{The price $B$ (black), and the opinions $X$ (red), and $Y$ (blue) on a logarithmic scale evolving in time.}
\label{fig:contSystem}
\end{center}
\end{figure}

We will be interested into the evolution of the distance between $X$ and $Y$. In illiquid markets, i.e. in markets wherein there is only few supply, already smaller demands can only be satisfied in connection with a strong change of the price. Thus, a large group of traders willing to trade for a certain price provides some resistance against further evolution of the price into this direction. Consequently, it is of great interest how longterm investors adapt to strong price changes since they are providing resistance on levels which are normally on some distance from the price. If these investors react to slowly, the price can fluctuate between these levels without much resistance, leading to strong volatility. The theory of trading strategies on illiquid markets is a very active field of research and there are many different approaches to model these markets and their reactions on trading \cite{almgren01, huberman05, obizhaeva05}. However, the question if large orders on illiquid markets can destabilize them, seems to be open.

{\it Bovier et al.} describe in \cite{bovier06} a class of Markovian agent-based models for the evolution of a share price. Therein they present the idea of a {\it virtual order book}, which keeps track of the trader's opinions about the value of the share, irrespective of whether they have placed an order or not. For practical purposes the model is stated in a discrete time setting and in every round one agent updates his opinion. As a main feature, the probability to be chosen depends on the distance of the agent to the price. In particular, in a market with $N$ traders and current price $p$ the probability for agent $i$ with current opinion $p_i$ to be chosen is given by
\begin{equation}
\frac{h(|p_i-p|)}{\sum_{j=0}^N h(|p_j-p|)}.
\end{equation}
The function $h$ is assumed to be positive and decreasing, reflecting the idea that traders with opinions far away from the price react slower to price changes. The model is stated in a very general setting, but the authors are able to reproduce on a qualitative level several statistical properties of the price process, sometimes called {\it stylized facts}, by choosing
\begin{equation}
h(x) = \frac{1}{\left(1+x\right)^\gamma}.
\end{equation}
We pick up on this choice for our model. The logarithmic price process $B$ will be a Brownian motion, whereas the opinions of buyers, $X$, and sellers, $Y$, are described by ordinary differential equations in dependence on parameter $\gamma > 0$ and the Brownian motion $B$.

The buyers opinion at time $t$ is given by the solution of
\begin{equation}\label{eq:odeIntro}
\frac{d}{dt}f(t)=\frac{1}{\left(1+B_t-f(t)\right)^\gamma},
\end{equation}
whenever $X_t < B_t$. By the argumentation above that the individuals within the group may change, but the group of traders itself remains, $X$ can hit $B$, but it is not allowed to cross it, and thus, it describes the same movement as $B$, until $B$ goes up so fast that it cannot follow (observe that $1$ is an upper bound for the speed of $X$). This happens immediately after the two processes have met, because $B$ is fluctuating almost everywhere. As soon as the distance is positive, $X$ is driven by (\ref{eq:odeIntro}) again. Since $B$ is differentiable almost nowhere, some work is needed to give a rigorous construction of this process.

For the opinion of shareholders $Y_t$ we assume the same construction with a changed sign on the right hand side of (\ref{eq:odeIntro}). $-B$ is also a Brownian motion, and thus we can define equivalently
\begin{equation}
Y(B)=-X(-B).
\end{equation}

Notice that the speed of adaption to price fluctuations is governed by the parameter $\gamma$ in our model. Therefore we are interested in  the longterm behaviour of $Y-X$ as a function of $\gamma$. In particular, we would like to know, when $Y-X$ is recurrent, and when it is transient. A heuristic argument suggests that $\gamma = 1$ is a critical value. For a constant, $c>0$, we scale time by $c^2$ and space by $c$. We denote the scaled versions of the processes by adding superscript $c$. By Brownian scaling we have that $B^c$ is equal to $B$ in distribution. On the other hand, $X^c$ solves
\begin{equation}
\frac{d}{dt}X^c_t=\frac{c^{1-\gamma}}{\left(1/c+B^c_t-X^c_t\right)^\gamma}.
\end{equation}
If one assumes $B^c_t-X^c_t$ to be larger than $0$, the slope tends to infinity for $\gamma < 1$, and to $0$ for $\gamma > 1$ as $c$ becomes large. This observation suggests that $Y-X$ remains stable for $\gamma < 1$, only. In this paper we  show that this first guess is right, and prove a rigorous statement about the stability in dependence on $\gamma$.

The remainder of this article is organized as follows. In Section \ref{sec:construction} we define the particle system formally. $X$, or $Y$ respectively, will be constructed pathwisely as a sequence of processes. The existence of these limits is stated in Lemma \ref{lem:limExists}, its lengthy proof is given in Appendix \ref{sec:appendix}. In Section \ref{sec:thmproof} we present the main theorem and its proof, and in Section \ref{sec:conclusions} we will discuss, what our results mean for the {\it opinion game} from \cite{bovier06}.

\section{Construction}\label{sec:construction}
We introduce the processes $B$, $X$, and $Y$ formally. While $B=(B_t)_{t\in\mathbb{R}_0^+}$ is just a Brownian motion on a probability space $\{\Omega,\mathscr{F},(\mathscr{F}_t)_{t\in\mathbb{R}_0^+},P\}$, $X$ is constructed pathwisely by introducing a sequence of random step functions $B^\epsilon(\omega)$, for which the distance to $B(\omega)$ is uniformly smaller or equal than $\epsilon$. The construction of $X^\epsilon$, attracted to $B^\epsilon$ in the sense as explained in the introduction, turns out to be easy. At last, we show that $X^\epsilon$ has a limit as $\epsilon$ tends to zero, and call this limit process $X$. The construction of $Y$ follows immediately afterwards. The advantage of a step function approach is the simple transition to a discrete setting, which we will use extensively in the proof of the main theorem later on.

For any $\epsilon > 0$ we define jump times by $\bar{\sigma}^\epsilon_0 := 0$ and
\begin{equation}\label{eq:sigmaBar}
\bar{\sigma}^\epsilon_i := \min\left\{t>\bar{\sigma}^\epsilon_{i-1}:\left|B_t-B_{\bar{\sigma}^\epsilon_{i-1}}\right|\geq\epsilon\right\};\ i\in\mathbb{N},
\end{equation}
neglecting the $\epsilon$-index, provided no confusion is caused. Furthermore, we define step functions $B^\epsilon:[0,\infty)\to\mathbb{R}$ by
\begin{equation}
B^\epsilon_t := B_{\bar{\sigma}_i}\textrm{ for }t\in \left[\bar{\sigma}_i,\bar{\sigma}_{i+1}\right).
\end{equation}
Observe, by definition
\begin{equation}
\sup_{t\geq 0}\left|B_t-B^\epsilon_t\right|=\epsilon\textrm{ a.s.},
\end{equation}
and thus $B^\epsilon$ converges to $B$ on the whole $\left[0,\infty\right)$ in $\sup$-norm.
As already mentioned in the introduction, we basically want $X$ to fulfil
\begin{equation}\label{eq:ode}
\frac{d}{dt}X_t = \left(1+B_t-X_t\right)^{-\gamma},
\end{equation}
as long as $X_t < B_t$. If we substitute $B$ by a fixed number $b\geq 0$, the {\it ode} (\ref{eq:ode}) is explicitly solvable. The solution of
\begin{equation}\label{eq:odeConst}
\frac{d}{dt}f(t)=\left(1+b-f(t)\right)^{-\gamma};\ f(0)=0
\end{equation}
is
\begin{equation}\label{eq:hFuncBar}
\bar{h}(t, b) := b+1-\left(\left(b+1\right)^{\gamma+1}-\left(\gamma+1\right)t\right)^{\frac{1}{\gamma+1}}.
\end{equation}
We will call $\bar{h}(t,b)$ {\it well-defined} if
\begin{equation}
b\geq 0\textrm{ and }t\leq \frac{(b+1)^{\gamma+1}-1}{\gamma+1}.
\end{equation}
Obeserve that the bound on $t$ ensures $\bar{h}(t,b)\leq b$. As we will be mainly interested in the distance of $\bar{h}$ to $b$ at time $t$, we set
\begin{equation}\label{eq:hFunc}
h(t,b):=\left\{\begin{array}{ll}
b-\bar{h}(t,b)&\textrm{if }\bar{h}(t,b)\textrm{ is well-defined}\\
0&\textrm{else}
\end{array}\right. .
\end{equation}
This motivates to define $X^\epsilon$ in the following way: For $t\in [\bar{\sigma}_i, \bar{\sigma}_{i+1})$, $i\in\mathbb{N}_0$, we set
\begin{equation}
X^\epsilon_t :=B^{\epsilon}_{\bar{\sigma}_i}-h(t-\bar{\sigma}_i,B^{\epsilon}_{\bar{\sigma}_i}-X^\epsilon_{\bar{\sigma}_i-}),
\end{equation}
whereby $X^\epsilon_{0-} := 0$ (Figure \ref{fig:epsSystem}). This means, for $t\in [\bar{\sigma}_i, \bar{\sigma}_{i+1})$ we first consider $X^\epsilon_{\bar{\sigma}_i-}$. If $B^\epsilon_{\bar{\sigma}_i}$ is smaller than this value, we set $X^\epsilon_t := B^\epsilon_{\bar{\sigma}_i}$. Else we can apply function $\bar{h}$ to calculate the movement of $X^\epsilon$ torwards $B^\epsilon$. If $X^\epsilon$ reaches $B^\epsilon$ before time $t$, it remains on this level.

\begin{figure}[htb]
\begin{center}
\includegraphics[width=.75\textwidth]{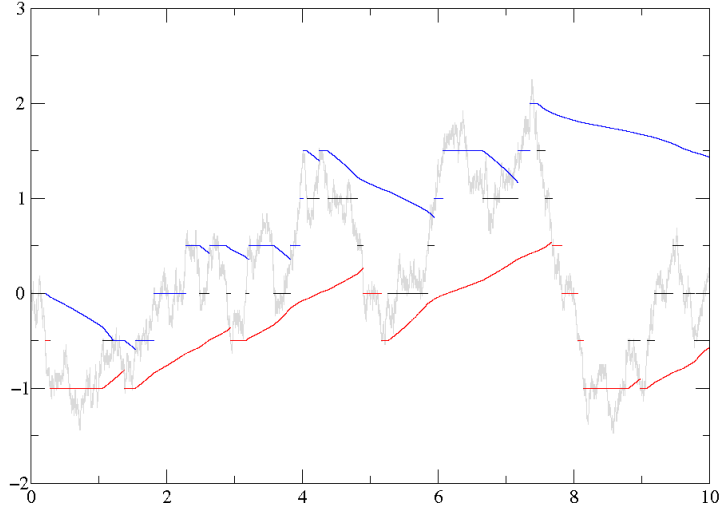}
\caption{The three processes $B^\epsilon$ (black), $X^\epsilon$ (red), and $Y^\epsilon$ (blue). $B$ is displayed beneath in grey. To make the construction clear, $\epsilon$ is chosen {\it large} in this figure ($\epsilon = 1/2$).}
\label{fig:epsSystem}
\end{center}
\end{figure}

\begin{lemma}\label{lem:limExists}
Let $S\subset [0,\infty)$ be a compact set and $\epsilon \ll \exp(-\gamma\cdot \sup S)$. Then  
\begin{equation}
\sup_{t\in S}\left|X^{\epsilon'}_t-X^\epsilon_t\right| \leq \epsilon K_S\textrm{ a.s.,}
\end{equation}
whereby $K_S$ is a finite, deterministic constant depending on $S$, and $\epsilon' < \epsilon$.
\end{lemma}
\begin{proof}
See appendix \ref{sec:appendix}.
\end{proof}
Lemma \ref{lem:limExists} shows that $(X^\epsilon_t)_{\epsilon>0}$ is a Cauchy sequence in the set of all bounded functions from $S$ to $\mathbb{R}$, equipped with the $\sup$-norm. As this space is complete, $(X^\epsilon_t)_{\epsilon>0}$ converges. We denote the limit process by $X$. Equivalently, we define
\begin{equation}
Y^\epsilon(B^\epsilon(\omega)) := -X^\epsilon(-B^\epsilon(\omega))\textrm{ and }Y(B(\omega)) := -X(-B(\omega)).
\end{equation}

\section{The main theorem}\label{sec:thmproof}
\subsection{The theorem}
\begin{thm}\label{thm:main}
Let $B$, $X$, and $Y$ be defined as before and let
\begin{equation}
\theta_r := \sup\left\{t\geq 0: \left|Y_t-X_t\right|\leq r\right\}
\end{equation}
be the last exit time from an $r$-ball with respect to the $||\cdot||_1$-norm. Then
\begin{enumerate}
\item for $\gamma < 1$
\begin{equation}
\left(\forall r>0\right)\ \theta_r = \infty\textrm{ a.s.,}
\end{equation}
\item and for $\gamma > 1$
\begin{equation}
\left(\forall r>0\right)\ \theta_r < \infty\textrm{ a.s.}
\end{equation}
\end{enumerate}
\end{thm}
The theorem confirms our guess concerning $1$ being a critical value for $\gamma$. For the critical case there is no statement at all, but as the proof of transience in the supercritical case seems to be sharp, our conjecture is null-recurrence if $\gamma = 1$.

We prove Theorem \ref{thm:main} by discretising the process $Y-X$. This results in a Markov chain, which we will examine in detail in Subsection \ref{sec:discrete}. In \ref{sec:subcrit} we prove the subcritical case by reducing it to a one-dimensional random walk problem. For the transient case ($\gamma > 1$) we basically use that a Markov chain is transient if we can find a bounded subharmonic function with respect to the generator of the chain. The particular theorem and its application in the proof can be found in Subsection \ref{sec:supercrit}.

\subsection{Discretising the problem and facts about Markov chains}\label{sec:discrete}
Let us look at the problem from another perspective. We consider the two-dimensional process $(B^\epsilon-X^\epsilon, Y^\epsilon-B^\epsilon)$, and interprete it in the following as particle moving in $[0,\infty)^2$. Observe that $Y^\epsilon-X^\epsilon$ is just the sum of both coordinates. Furthermore, because $Y^\epsilon-X^\epsilon$ can only increase at times $\bar{\sigma}_i$ and decreases afterwards, we have
\begin{equation}\label{eq:locMin}
\inf_{t\in[\bar{\sigma}_i,\bar{\sigma}_{i+1})} \left(Y^\epsilon-X^\epsilon\right)_t = \left(Y^\epsilon-X^\epsilon\right)_{\bar{\sigma}_{i+1}-}.
\end{equation}
For all $\epsilon > 0$ we define a two-dimensional Markov chain $\Phi^\epsilon = \Phi(B^\epsilon) := \left(\Phi(B^\epsilon)_i\right)_{i\in\mathbb{N}}$ with state space $[0,\infty)^2$, equipped with the Borel-$\sigma$-algebra $\mathfrak{B}([0,\infty)^2)$, by
\begin{equation}
\Phi^\epsilon_i := \left(B^\epsilon-X^\epsilon,Y^\epsilon-B^\epsilon\right)_{\bar{\sigma}_i-},
\end{equation}
whereby $\bar{\sigma}_0-=0$. The $j$-step transition probabilities from $x\in [0,\infty)^2$ to $A\subset [0,\infty)^2$ will be denoted by $P^j_x(A)$, neglecting the index for $j=1$, and the generator $L$ will be given by
\begin{equation}
Lg(x) := \int_{[0,\infty)^2} P_x(dy)g(y)-g(x)
\end{equation}
for suitable functions $g:[0,\infty)^2\to [0,\infty)$.

\label{sec:finerStruc} In the following it will be of great importance to understand, how the particle moves exactly, while $\Phi^\epsilon_i = (x,y)$ jumps to $\Phi^\epsilon_{i+1}$ (Figure \ref{fig:discSystem}). At first, a jump of size $\epsilon$ happens at time $\bar{\sigma}_i$. The position afterwards is either $(x+\epsilon,\ (y-\epsilon)\ \vee\ 0)$ or $((x-\epsilon)\ \vee\ 0,\ y+\epsilon)$ with probability $1/2$ each. Let us call this new position $(x',y')$. Before the next jump happens at time $\bar{\sigma}_{i+1}$, the particle drifts in direction of the origin. If it reaches one of the axis, it remains there, and only drifts torwards the other one, until it has reached $(0,0)$. Thus, the coordinates of $\Phi^\epsilon_{i+1}$ are given by $(h(\bar{\sigma}_{i+1}-\bar{\sigma}_i, x'),\ h(\bar{\sigma}_{i+1}-\bar{\sigma}_i, y'))$. Observe that $\Phi^\epsilon$ can only increase (in $||\cdot||_1$-sense) on the axes. 

\begin{figure}[htb]
\begin{center}
\psfrag{X}{$B^\epsilon-X^\epsilon$}
\psfrag{Y}{$Y^\epsilon-B^\epsilon$}
\includegraphics[width=.5\textwidth]{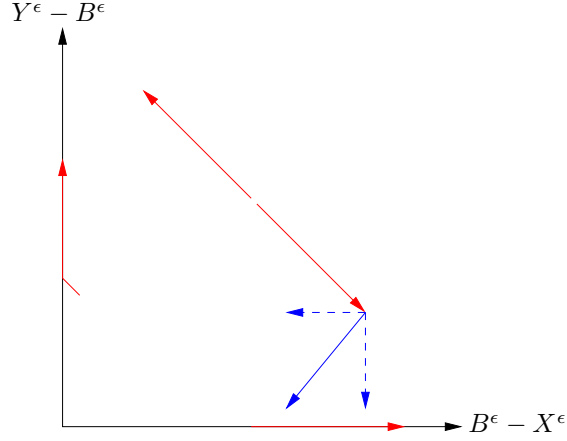}
\caption{The particle jumps (red arrows) parallel to the level lines of the $||\cdot ||_1$-norm. In the sense of this norm it can only increase on the axes. The drift consists of two independent drift components (blue dashed arrows), orthogonal to the axes. The resulting drift is illustratetd by the solid blue arrow.}
\label{fig:discSystem}
\end{center}
\end{figure}

Next, we need to understand the distribution of $\bar{\sigma}_{i+1}-\bar{\sigma}_i$. Thus, we set
\begin{equation}\label{eq:sigma}
\sigma_i := \bar{\sigma}_{i+1}-\bar{\sigma}_{i} \stackrel{d}{=} \inf\left\{t>0:B_t = \epsilon\right\}.
\end{equation}
As already suggested in the equation above, all $\sigma_i$ are {\it i.i.d} with support on $(0,\infty)$ and $\mathbb{E}\sigma = \epsilon^2$. The distribution is not known explicitly, but it can be expressed as a series with alternating summands with decreasing absolute values (refer to section C.2 in \cite{berglund06}). Calculating the first two summands results in
\begin{eqnarray}
&&\frac{4}{\pi}e^{-\pi^2/(8\epsilon)}\left(1-\frac{1}{3}e^{-\pi^2/\epsilon}\right)\\
&\leq&P\left(\sigma>\epsilon\right)\ =\ P\left(\sup_{0\leq s\leq \epsilon} \left|B_s\right|<\epsilon\right)\\
&\leq&\frac{4}{\pi}e^{-\pi^2/(8\epsilon)}.
\end{eqnarray}
For our purposes it will be sufficient to know that both bounds are of order $\exp(-1/\epsilon)$.

As we are operating on a continuous state space, the question for irreducibilty is a question for reaching sets instead of single states. Formally, $\Phi^\epsilon$ is called {\it $\varphi$-irreducible} if there exist a measure $\varphi$ on $\mathfrak{B}([0,\infty)^2)$ s.th.
\begin{equation}
\varphi(A)>0\Rightarrow P_x\left(\Phi^\epsilon\textrm{ ever reaches }A\right) > 0\textrm{ for all }x\in [0,\infty)^2.
\end{equation}
In our case
\begin{equation}
P_x\left(\left\{{\bf 0}\right\}\right) > 0 \textrm{ for all }x\in [0,\infty)^2,
\end{equation}
because the support of $\sigma$ is unbounded. Thus $\Phi^\epsilon$ is $\delta_{\bf 0}$-irreducible. The existence of an irreducibility measure ensures that there is also a {\it maximal irreducibility measure} $\Psi$ (compare with \cite{meyn96}, Prop. 4.2.2) on $\mathfrak{B}([0,\infty)^2)$ with the properties:
\begin{enumerate}
\item $\Psi$ is a probability measure.
\item $\Phi^\epsilon$ is $\Psi$-irreducible.
\item $\Phi^\epsilon$ is $\varphi'$-irreducible iff $\Psi \succ\varphi'$ (i.e. $\Psi(A)=0 \Rightarrow \varphi'(A)=0$).
\item $\Psi(A) = 0\ \Rightarrow\ \Psi\left(\left\{x:P_x(\Phi^\epsilon\textrm{ ever enters }A)\right\}\right)=0$.
\item In our case $\Psi$ is equivalent to
\begin{equation}
\Psi'(A) = \sum_{j=0}^\infty P^j_0(A)2^{-j}.
\end{equation}
\end{enumerate}
We denote the set of measurable, $\Psi$-irreducible sets by
\begin{equation}
\mathfrak{B}^+([0,\infty)^2):= \{A\in\mathfrak{B}([0,\infty)^2):\Psi(A)>0\}.
\end{equation}

Because the density of $\bar{\sigma}_{i+1}-\bar{\sigma}_i$ has support on $(0,\infty)$, it is not hard to see that
\begin{equation}
\mu(A):=Leb(A)+\delta_{\bf 0}(A)\neq 0\ \Rightarrow\ \Psi(A)\neq 0,
\end{equation} 
and therefore $\Psi \succ \mu$, whereby $Leb$ denotes the Lebesgue measure. 

Since $\Phi^\epsilon$ is a Markov chain on the {\it (possible) local minima} of $Y^\epsilon-X^\epsilon$ in the sense of (\ref{eq:locMin}), it is obvious that transience of $\Phi^\epsilon$ implies transience of $Y^\epsilon-X^\epsilon$. On the other hand $||\Phi^\epsilon||_1$ can only increase by at most $\epsilon$ in every step. Thus,
\begin{equation}
\sup_{t\in[\bar{\sigma}_i,\bar{\sigma}_{i+1})} \left(Y^\epsilon-X^\epsilon\right)_t \leq \left|\left|\Phi^\epsilon_i\right|\right|_1 +\epsilon
\end{equation}
and recurrence of $\Phi^\epsilon$ also implies recurrence of $Y^\epsilon-X^\epsilon$. However, observe that the proof of recurrence/transience for $Y^\epsilon-X^\epsilon$, $\epsilon > 0$, would not directly imply recurrence/transience for $Y-X$ in general, because we only have convergence on compact sets. Thus, we will shortly argue in the end of both parts of the proof, why the desired result follows in our case.

\subsection{Proof of the subcritical case: $\gamma < 1$}\label{sec:subcrit} 
For the subcritical case we reduce the movement of $\Phi^\epsilon$ to a nearest neighbour random walk on the level sets
\begin{equation}
M(k) := \left\{\left(x,y\right)\in[0,\infty)^2\left|\ x+y = 4^k\right.\right\},\ k\in\mathbb{Z},
\end{equation}
of $||\cdot ||_1:[0,\infty)^2\to[0,\infty)$, and show that the probability to jump to $M(k-1)$ is larger than $1/2+\delta$, $\delta > 0$, for small $\epsilon$ and all $k \geq k^*$ for a $k^*\in \mathbb{Z}$. Then it is well-known that $||\Phi^\epsilon||_1<4^{k^*}$ infinitely often. Recurrence for $\Phi^\epsilon$ follows by irreducibility.

In particular, we introduce for $k\in\mathbb{Z}$
\begin{equation}
M^-(k) := \left\{\left(x,y\right)\in[0,\infty)^2\left|\ x+y\leq 4^{k-1}\right.\right\},
\end{equation}
\begin{equation}
M^+(k) := \left\{\left(x,y\right)\in[0,\infty)^2\left|\ x+y\geq 4^{k+1}\right.\right\},
\end{equation}
and the hitting time of $\Phi^\epsilon$ for a set $M\subseteq [0,\infty)^2$
\begin{equation}
\tau^\epsilon_M := \min \left\{i:\Phi^\epsilon_i\in M\right\},
\end{equation}
neglecting the $\epsilon$ whenver possible. Then we have to show
\begin{equation}
\left(\exists k^*\right)\left(\forall k\geq k^*\right)\ \lim_{\epsilon\to 0}\inf_{m\in M(k)} P_m\left(\tau_{M^-(k)}^\epsilon < \tau_{M^+(k)}^\epsilon\right)> 1/2+\delta,\ \delta>0.
\end{equation}

The proof works in four steps (Figure \ref{fig:proofIrred}).
\begin{enumerate}
\item We show that $P_m(\tau_{M^-(k)} < \tau_{M^+(k)})$ is minimized for $m^*\in\{(4^k,0),(0,4^k)\}$. As the model is symmetric, {\it w.l.o.g.}, we may assume $m^* = (0,4^k)$. 
\item We show
\begin{equation}
P_{m^*}\left(\tau_{\{(x,y):x=y\}} < \tau_{M^+(k)}\right) > 1-e^{-6/7} \approx 0.576
\end{equation}
as $\epsilon$ tends to $0$.
\item We assume, the particle has been successful in the last step, and has reached $(x,x)\notin M^+(k)$. Then, in the worst case, it is at position $(2\cdot4^k,2\cdot4^k)$ now, respectively arbitrarily close to it (as $\epsilon$ becomes small). As the direction of the jumps and the drift times $\sigma_i$ are mutually indpendent, we can treat jump and drift phases independently. We will use this knowledge to determine the diameter of a tube around the bisector. As long as the particle is located within this area, it will not drift to the axis {\it too fast}. When we know the diameter, we can calculate the probability that the jumps do not take the particle out of the tube within a certain time period. Knowing this time and the speed torwards the origin, we can calculate, how close it gets to the origin before hitting the axis.
\item Finally, we combine steps $2$ and $3$. It will turn out that the probability to stay in the tube for a certain time (step $3$) can be chosen large enough such that it is still strictly larger than $1/2$ if multiplied with the probability to reach the bisector (step $2$). On the other hand, the time $\Phi^\epsilon$ stays in the tube, will be sufficient to reach $M^-(k)$.
\end{enumerate} 

\begin{figure}[htb]
\begin{center}
\psfrag{X}{$B^\epsilon-X^\epsilon$}
\psfrag{Y}{$Y^\epsilon-B^\epsilon$}
\psfrag{k}{$M(k)$}
\psfrag{-}{$M^-(k)$}
\psfrag{+}{$M^+(k)$}
\includegraphics[width=.7\textwidth]{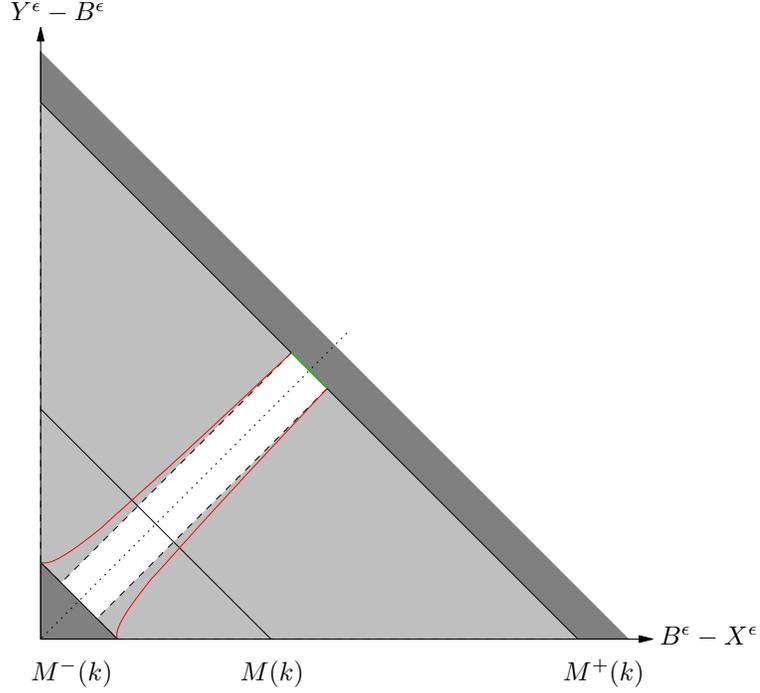}
\caption{The idea of the proof: The particle starts in $M(k)$ (black line). We show that the probability to get to $M^-(k)$ (left dark grey area) before it gets to $M^+(k)$ (right dark grey area) is larger than $1/2$. This probability is bounded from below by the product of the probability to reach the bisector (dotted line) before reaching $M^+(k)$, and the probability to get back to $M^-(k)$ before hitting the axis. In particular, we calculate the probability of a random walk with step size $\sqrt{2}\epsilon$ to stay in the white slot around the bisector. Its diameter (green line) $diam(A_{4^{k+1}})$ is a lower bound for the diameter of the area enclosed by $g(x)$ and $g^{-1}(x)$ (red lines).}
\label{fig:proofIrred}
\end{center}
\end{figure}

For step $1$ we consider a realisation $B^\epsilon(\omega)$ of the Brownian step function, and $X^\epsilon(B^\epsilon(\omega))$ attracted to this realisation with starting distance $|B^\epsilon_0-X^\epsilon_0| = d$, as well as $\bar{X}^\epsilon(B^\epsilon(\omega))$, constructed like $X^\epsilon$ and attracted to the same realisation, but with initial distance $|B^\epsilon_0-\bar{X}^\epsilon_0| = \bar{d}$, $\bar{d} > d$. Because $X^\epsilon$ is Markovian, we can easily extend our construction of $X^\epsilon$ to initial values different from $0$. Then
\begin{equation}
\left(\forall t\geq 0\right)\left[\left(B^\epsilon(\omega)-\bar{X}^\epsilon\right)_t \geq \Big(B^\epsilon(\omega)-X^\epsilon\Big)_t\right]
\end{equation}
with equality for all $t\geq r \geq 0$, whereby $r$ fulfils
\begin{equation}
\left(B^\epsilon(\omega)-\bar{X}^\epsilon\right)_r = 0 = \Big(B^\epsilon(\omega)-X^\epsilon\Big)_r.
\end{equation}
By symmetry the respective statement holds also for $Y^\epsilon-B^\epsilon$. Thus, if $\Phi^\epsilon_i(\omega)$ is smaller or equal in both coordinates than a copy $\bar{\Phi}^\epsilon_i(\omega)$ for some time i, this (in)equality will remain for all times afterwards. Thus, we can conclude that for $x < x'$
\begin{equation}\label{eq:irredPart1}
P_{(x,0)}\left(\tau_{M^-(k)} < \tau_{M^+(k)}\right)\geq P_{(x',0)}\left(\tau_{M^-(k)} < \tau_{M^+(k)}\right),
\end{equation}
because every realisation of $B^\epsilon$ fulfiling the event on the right side also fulfils the one on the left side.

As $\Phi^\epsilon$ can only increase at the axes, starting it from a point inside the quadrant will result in a decrease of both coordinates until one of the axes is hit. But then (\ref{eq:irredPart1}) applies, and therefore step $1$ is proven.\\
\ \\
In step $2$ we show
\begin{equation}
P_{m^*}\left(\tau_{\{(x,y):x=y\}} > \tau_{M^+(k)}\right) < e^{-6/7}.
\end{equation}
We assume $m^*=(0,4^k)$. The particle has two possibilities now. Either it jumps upwards the axis to $(0,4^k+\epsilon)$, or it jumps into the quadrant to $(\epsilon, 4^k-\epsilon)$. Afterwards it drifts. In this step we will ignore the drift phase for two reasons. First, the change of position by jumps is of order $\epsilon$, while it is of order $\epsilon^2$ by drifting, because $\mathbb{E}\sigma = \epsilon^2$. Furthermore, the drift direction is different from the jump direction, and for every change of position in jump direction by drifting, there is also a drift {\it down}, orthogonal to the jump direction, by the same amount at least. Thus, considering the drift would help us in reaching our aim to drift down.

We introduce the following {\it game}: sitting on the axis, the particle can either reach the bisector or it can move up the axis by $\epsilon$. As the particle needs $4^k/(2\epsilon)$\footnote{Here we neglect that the expression is meaningful for integers only, because the difference will not play a role as $\epsilon$ tends to zero.} steps to reach the bisector, but only one step to go up, the success probability is small. If we should not success, we have another chance at $(0,4^k+\epsilon)$ (even if the probability for success is smaller there) and so on, until we reach $M^+(k)$. It is well known that the probability of an one-dimensional, symmetric random walk to reach $-1$ before it reaches $k\in\mathbb{N}$, started in ${0}$, is given by $k/(k+1)$. Thus,
\begin{eqnarray}
P_{m^*}\left(\tau_{\{(x,y):x=y\}}> \tau_{M^+(k)}\right)&=&\prod_{i = 0}^{\left(4^{k+1}-4^k\right)/\epsilon-1}\frac{(4^k/2+i\epsilon)/\epsilon}{(4^k/2+i\epsilon)/\epsilon+1}\\
&=&\prod_{i = 0}^{3\cdot 4^k/\epsilon-1}\frac{4^k/2+i\epsilon}{4^k/2+i\epsilon+\epsilon}\\
&<&\left(\frac{4^k/2+3\cdot 4^k-\epsilon}{4^k/2+3\cdot 4^k}\right)^{3\cdot 4^k/\epsilon}\\
&=&\left(1-\frac{2\epsilon}{7\cdot 4^k}\right)^{3\cdot 4^k/\epsilon}\\
&\to& e^{-6/7}\textrm{ as $\epsilon$ tends to $0$.}
\end{eqnarray} 
\ \\
For part $3$ we assume that the particle has reached the bisector and is at position $(2\cdot 4^k,2\cdot 4^k)$. First, we are interested in the speed of the particle while drifting. In particular, we are looking for a uniform lower bound for the speed orthogonal to the $||\cdot ||_1$-level sets on $[0,\infty)^2\backslash M^+(k)$. If we denote the particle's current position by $(x,y)$, its speed in $x$-direction is given by $(1+x)^{-\gamma}$ and in $y$-direction by $(1+y)^{-\gamma}$, because of equation (\ref{eq:odeConst}). Thus, the speed orthogonal to the level sets is given by
\begin{equation}
v_{(x,y)}:=\sqrt{\left(1+x\right)^{-2\gamma}+\left(1+y\right)^{-2\gamma}}.
\end{equation}
Differentiation of $v$ shows that on the set $\{(x,y):x+y\leq 4^{k+1}\}$ the speed is minimized exactly on position $(2\cdot 4^k,2\cdot 4^k)$ and amounts
\begin{equation}\label{eq:minSpeed}
v_{\min} := \sqrt{2}\left(1+2\cdot 4^k\right)^{-\gamma}.
\end{equation}
Next, let us take a closer look at the movement of the particle while drifting. Observe first that a drifting particle started in $(x,y)$ will never cross the path of a second particle, started somewhere else, before it has hit one of the axes. This follows directly from our argumentation in step $1$. Let us assume that $x\leq y$. By symmetry the other case will follow immediately. In this case, the particle will first hit the $x$-axes, and that happens at time
\begin{equation}
t_x := \min \{t:h(x,t)=0\} = \frac{\left(x+1\right)^{\gamma +1}-1}{\gamma +1},
\end{equation}
which follows from the definition of $h$ in (\ref{eq:hFunc}). What constraints must hold for $y$ such that the particle will hit the axes in $M^-(k)$? Clearly, $y$ must fulfil
\begin{equation}
h(y,t_x)\leq 4^{k-1}\textrm{, or equivalently}
\end{equation}
\begin{equation}\label{eq:gFunc}
y\leq \left(\left(4^{k-1}+1\right)^{\gamma + 1}+\left(x+1\right)^{\gamma +1}-1\right)^{1/(\gamma + 1)}-1.
\end{equation}
Let us denote the right side of the last inequality by $g(x)$. By differentiation we immediately see that $g(x)-x$ is a positive, strictly decreasing function, tending to $0$ as $x$ becomes large. On the other hand, for starting position $(x,y)$, $y\leq x$, the calculation would be the same with exchanged roles of $y$ and $x$, and we would end up with $g(y)$. Thus, as long as the particle starts in
\begin{equation}
(x,y)\in A:=\left\{(x,y):\left(x+y \leq 4^{k+1}\right) \wedge \left(g^{-1}(x)\leq y\leq g(x)\right)\right\},
\end{equation}
it will first reach $M^-(k)$ and hit the axis only afterwards. This leads to the crucial observation: as long as the particle only jumps to positions $(x,y)\in A$, we do not have to worry that the particle will reach the axis before reaching $M^-(k)$.

Let us define the level sets of $A$ by
\begin{equation}
A_l := A\cap \left\{(x,y):x+y=l\right\}.
\end{equation}
We can interprete $A_l$ as a one dimensional interval or a piece of a line, and because $g(x)-x$ and $g^{-1}(x)-x$ are tending to zero, the length of this interval, denoted by $diam(A_l)$, decreases as $l$ increases. Thus, we would like to know $diam(A_{4^{k+1}})$, as it is a lower bound for all $l$ we are interested in. Because the jump direction of the particle is parallel to the $A_l$, we can afterwards estimate, how much time the particle will spend in $A$ when performing jumps. However, it is not possible to calculate $diam(A_{4^{k+1}})$ explicitly, but by Pythagorean Theorem, the symmetry of $g(x)$ and $g^{-1}(x)$, as well as the decrement of $g(x)-x$ again, we have
\begin{equation}
diam(A_{4^{k+1}}) \geq \sqrt{2}\left(g(2\cdot 4^k)-2\cdot 4^k\right) =: d_k.
\end{equation}
Our ansatz is
\begin{equation}
d_k \geq D4^k
\end{equation}
for a constant $D$, independent of $k$ if $k$ is large enough. Notice that function $g$ as defined in (\ref{eq:gFunc}) is basically the $||\cdot ||_{\gamma +1}$-norm of $(4^{k-1},x)$ and decreases in $\gamma$. {\it W.l.o.g.}, we may assume that $\gamma = 1$.
\begin{equation}
\sqrt{2}\left(\left(\left(4^{k-1}+1\right)^2+\left(2\cdot 4^k+1\right)^2-1\right)^{1/2}-1-2\cdot4^k\right)\geq D4^k
\end{equation}
easily transforms to
\begin{equation}
\sqrt{2}D\leq \left(\frac{\sqrt{65}-8}{2}-O(4^{-k})\right).
\end{equation}
Finally, we have to answer the question, how long do we remain in an interval of diameter $\sqrt{2}D4^k$, when we start in the centre and perfom a random walk with step size $\sqrt{2}\epsilon$. Let us denote a standard random walk with step size $1$ by $R$, then we are looking for the hitting time
\begin{equation}
\xi^\epsilon(k):=\min\left\{n:R_n\notin (-D4^k/\epsilon,D4^k/\epsilon)\right\}.
\end{equation}
It is well known that
\begin{equation}
\mathbb{E}\xi^\epsilon(k) = \left(\frac{D4^k}{\epsilon}\right)^2.
\end{equation}
We would like to have a lower bound for the probability that we stay in the interval for $c\mathbb{E}\xi^\epsilon(k)$ steps at least, whereby $c\in(0,1)$ can be arbitrarily small. It will be sufficient to show that this probability tends to $1$ if $c$ goes to $0$. As $\epsilon$ tends to zero, Donsker's principle (see chapter 2.4.D of \cite{karatzas91}) tells us that
\begin{equation}
\lim_{\epsilon \to 0}\frac{D4^k}{\epsilon} \tilde{R}_{(D4^k/\epsilon)^2t}\stackrel{d}{=}B_t,
\end{equation}
whereby $\tilde{R}$ is the linear interpolation of $R$. We define the exit time of a Brownian motion $B$ from $(-1,1)$ by
\begin{equation}
\bar{\xi} := \inf\left\{t:B_t \notin (-1,1)\right\}.
\end{equation}
If we use Donsker's principle we get for $\epsilon$ tending to zero and a constant $\alpha > 0$
\begin{eqnarray}
P\left(\xi^\epsilon(k) < c\mathbb{E}\xi^\epsilon(k)\right)&=&P\left(\bar{\xi}<c\right)\\
&=&P\left(\exp\left(-\alpha\bar{\xi}\right)>\exp\left(-\alpha c\right)\right)\\
&<&\frac{\mathbb{E}e^{-\alpha\bar{\xi}}}{e^{-\alpha c}}\label{eq:markov}\\
&=&\frac{e^{\alpha c}}{\cosh\left(\sqrt{2\alpha}\right)}\label{eq:borodin}.
\end{eqnarray}
In line (\ref{eq:markov}) we have used the Markov inequality, in line (\ref{eq:borodin}) the explicit formula for the Laplace transform of $\bar{\xi}$ (refer to formula $3.0.1$ in \cite{borodin96}). As $\alpha$ was chosen arbitrary, we would like to minimize line (\ref{eq:borodin}) as a function of $\alpha$. Differentiation shows that the optimizing $\alpha$ fulfils
\begin{equation}\label{eq:optAlpha}
\cosh\left(\sqrt{2\alpha}\right)=\frac{\sinh\left(\sqrt{2\alpha}\right)}{c\sqrt{2\alpha}}
\end{equation}
Using equality (\ref{eq:optAlpha}) in (\ref{eq:borodin}) results in
\begin{equation}
P\left(\xi^\epsilon(k) < c\mathbb{E}\xi^\epsilon(k)\right)<\frac{c\sqrt{2\alpha}e^{\alpha c}}{\sinh\left(\sqrt{2\alpha}\right)}
\end{equation}
which tends to zero as $c$ tends to zero. Let us call
\begin{equation}
p_c := P\left(\xi^\epsilon(k) \geq c\mathbb{E}\xi^\epsilon(k)\right)
\end{equation}
and observe that one can choose $c$ s.th. $p_c$ is arbitrarily close to one.\\ 
\ \\
In step $4$ we summarise the results from the steps before. When the particle starts in $M(k)$, the probability to reach the bisector, before it reaches $M^+(k)$ is larger than $1-\exp(-6/7)$ by steps $1$ and $2$. By step $3$ we can find a $c^*>0$ such that $(1-\exp(-6/7))p_{c^*} > 1/2$. This means, we will stay within $A$ for $c^*(D4^k/\epsilon)^2$ steps at least. As the particle drifts with a minimal speed $v_{\min}$, defined in (\ref{eq:minSpeed}), it will decrease its distance to the origin in terms of the $||\cdot ||_1$-norm by
\begin{eqnarray}
&&\sum_{i=1}^{c^*(D4^k/\epsilon)^2} \sqrt{2}\left(1+2\cdot 4^k\right)^{-\gamma}\sigma_i \\
&=&c^*D^2 4^{2k}\sqrt{2}\left(1+2\cdot 4^k\right)^{-\gamma}\label{eq:indepAlpha}\\
&=&O\left(4^{(2-\gamma)k}\right)\label{eq:order1}
\end{eqnarray}
for $\epsilon$ tending to zero. In line (\ref{eq:indepAlpha}) we have used the LLN for the {\it i.i.d.} $\sigma_i$, which have expectation $\epsilon^2$. Thus, the distance, the particle covers, is of order $4^{(2-\gamma)k}$. On the other hand, the distance, the particle has to cover to get to $M^-(k)$, is by construction of the proof smaller or equal than
\begin{equation}
4^{k+1}-4^{k-1} = O\left(4^k\right)\label{eq:order2}.
\end{equation}
Obviously (\ref{eq:order1}) dominates (\ref{eq:order2}) for $\gamma < 1$, which finishes the proof in the subcritical case for $Y^\epsilon-X^\epsilon$.\\
\ \\
To see that the result transfers to $Y-X$, we consider the process $\tilde{X}^\epsilon$, constructed in the same way like $X^\epsilon$ but with the modified {\it ode}
\begin{equation}\label{eq:modOdeConst}
\frac{d}{dt}f(t)=\left((1+2\epsilon)+b-f(t)\right)^{-\gamma};\ f(0)=0
\end{equation}
instead of the original {\it ode} (\ref{eq:odeConst}). Equivalently, we define $\tilde{Y}^\epsilon(B^\epsilon) := -\tilde{X}^\epsilon(-B^\epsilon)$. The proof easily shows that the change of the constant from $1$ to $1+2\epsilon$ in (\ref{eq:modOdeConst}) does not change the calculations or the result in an essential way (apart from longer equations because of the extra term $2\epsilon$). Thus, $\tilde{Y}^\epsilon-\tilde{X}^\epsilon$ is also recurrent for $\gamma < 1$. The crucial observation is that these {\it auxiliary} processes sandwich the original processes:
\begin{equation}
X^{\epsilon'}_t\geq\tilde{X}^\epsilon_t-\epsilon\textrm{ and }Y^{\epsilon'}_t\leq\tilde{Y}^\epsilon_t+\epsilon
\end{equation}
for all $\epsilon' < \epsilon$. This holds due to the fact that $|X^{\epsilon'}_{\bar{\sigma}_1^\epsilon}-X^{\epsilon}_{\bar{\sigma}_1^\epsilon}|<\epsilon$ and $|B^\epsilon-B^{\epsilon'}|<\epsilon$. Thus the difference of speed cannot be larger than $2\epsilon$. This argument extends inductively to all later times $\bar{\sigma}_i$. It follows
\begin{eqnarray}
Y_t-X_t&=&\lim_{\epsilon\to 0}(Y^\epsilon_t-X^\epsilon_t)\\
&\leq&\tilde{Y}^\epsilon_t-\tilde{X}^\epsilon_t+2\epsilon,
\end{eqnarray}
which proves recurrence for $Y-X$.

\subsection{Proof of the supercritical case: $\gamma > 1$}\label{sec:supercrit}
We first define, what transience of Markov chains means.
\begin{defi}
For any $A\subset [0,\infty)^2$ let
\begin{equation}
\eta_A := \sum_{i=0}^\infty \mathbbm{1}_{\left\{\Phi^\epsilon_i\in A\right\}}
\end{equation}
be the number of visits of $\Phi^\epsilon$ in $A$. A set $A$ is called {\it uniformly transient} if for there exists $M < \infty$ such that $\mathbb{E}_{(x,y)}(\eta_A) \leq M$ for all $(x,y)\in A$. We call $\Phi^\epsilon$ {\it transient} if there is a countable cover of $[0,\infty)^2$ with uniformly transient sets.
\end{defi}
We will use the next theorem to show that $\Phi^\epsilon$ is transient in the upper sense. It is stated as a more general result in \cite{meyn96}, 8.0.2(i).
\begin{thm}\label{thm:meyn}
The chain $\Phi^\epsilon$ is transient if and only if there exists a bounded, non-negative function $g:[0,\infty)^2\to [0,\infty)$ and a set $\mathcal{B} \in \mathfrak{B}^+([0,\infty)^2)$ such that for all $\left(\bar{x},\bar{y}\right)\in [0,\infty)^2\backslash \mathcal{B}$,
\begin{equation}\label{eq:toShowOrig}
Lg(\bar{x},\bar{y}) = \int_{[0,\infty]^2} P_{(\bar{x},\bar{y})}(d(x,y))g(x,y) \geq g(\bar{x},\bar{y})
\end{equation}
and
\begin{equation}
D:=\left\{(x,y)\in [0,\infty)^2\ \left|\ g(x,y)>\sup_{(\bar{x},\bar{y})\in \mathcal{B}} g(\bar{x},\bar{y})\right.\right\}\in\mathfrak{B}^+([0,\infty)^2).
\end{equation}
\end{thm}
Basically, we have to find a certain function $g$ such that we jump away from the origin in expectation with respect to $g$. This must hold outside a compact set $\mathcal{B}$ containing the origin. To find a proper $\mathcal{B}$ we set for all $z>0$  
\begin{equation}
\mathcal{B}_z := \left\{\left(x,y\right)\in[0,\infty)^2\ \big|\ \left|\left|\left(x+1,y+1\right)\right|\right|_{\gamma +1} = z\right\}.
\end{equation}
For $g$ we choose
\begin{equation}
g(x,y) := 1-\left|\left|\left(x+1,y+1\right)\right|\right|_{\gamma +1}^{-1}.
\end{equation}
If we can find a $\bar{z}$, remaining finite as $\epsilon$ tends to zero, s.th. equation (\ref{eq:toShowOrig}) holds for all $\left(x,y\right)\in \mathcal{B}_z,\ z\geq \bar{z}$, we are done. Recall what happens in one step of $\Phi^\epsilon$ in the underlying process, described on page \pageref{sec:finerStruc}. Equation (\ref{eq:toShowOrig}) becomes
\begin{eqnarray}\label{eq:toShowMod1}
\phantom{+}\frac{1}{2}\int_0^\infty P\left(\sigma\in dt\right)g(h(t,\bar{x}+\epsilon),h(t,\bar{y}-\epsilon))&&\\
+\frac{1}{2}\int_0^\infty P\left(\sigma\in dt\right)g(h(t,\bar{x}-\epsilon),h(t,\bar{y}+\epsilon))&\geq&g(\bar{x},\bar{y})\nonumber
\end{eqnarray}
whereby $(\bar{x},\bar{y})\in \mathcal{B}_{\bar{z}}$. Because of the $\epsilon$-jump of $B^\epsilon$ at time $\bar{\sigma}$, the integral splits into two parts. Within both integrals the only source of randomness is $\sigma$. Given its value, we can calculate the next position of $\Phi^\epsilon$ by using function $h$, and finally apply $g$ to this value.\\
Using the definition of $g$ and observing that the integral of the density $P(\bar{\sigma}\in dt)$ is one, (\ref{eq:toShowMod1}) easily transforms to
\begin{eqnarray}\label{eq:toShowMod2}
\phantom{+}\frac{1}{2}\int_0^\infty P\left(\sigma\in dt\right)\left|\left|(h(t,\bar{x}+\epsilon)+1,h(t,\bar{y}-\epsilon)+1)\right|\right|_{\gamma+1}^{-1}&&\\
+\frac{1}{2}\int_0^\infty P\left(\sigma\in dt\right)\left|\left|(h(t,\bar{x}-\epsilon)+1,h(t,\bar{y}+\epsilon))+1\right|\right|_{\gamma+1}^{-1}&\leq&\bar{z}^{-1}\nonumber
\end{eqnarray}
As already argued, $\sigma$ is small, or rather we can change the upper bound of the integrals from $\infty$ to $\epsilon$ at the expense of order $\exp(-1/\epsilon)$. Furthermore, let us assume for the moment that $\bar{x},\bar{y}\geq 2\epsilon$. As jump size and drift time are $\epsilon$ at most and the drift speed is bounded from above by $1$ this condition avoids that we have to handle cases in which the axes are reached. Observe that the only special cases to check later on are $(0,\bar{y})$ and $(\bar{x},0)$, because we can choose for every pair $\bar{x},\bar{y}>0$ an $\epsilon > 0$ such that the condition above is fulfiled, and we let $\epsilon$ tend to zero. Now we can use Taylor approximations for $\epsilon$ and $t$ to get
\begin{eqnarray}
&&\frac{1}{2}\left(\left|\left|(h(t,\bar{x}+\epsilon)+1,h(t,\bar{y}-\epsilon)+1)\right|\right|_{\gamma+1}^{-1}\right.\\
&&\phantom{\frac{1}{2}}\left.+\left|\left|(h(t,\bar{x}-\epsilon)+1,h(t,\bar{y}+\epsilon)+1)\right|\right|_{\gamma+1}^{-1}\right)\nonumber\\
&=&\frac{1}{2}\left(\left(\left(\bar{x}+\epsilon+1\right)^{\gamma+1}+\left(\bar{y}-\epsilon+1\right)^{\gamma+1}-2\left(\gamma+1\right)t\right)^{-\frac{1}{\gamma+1}}\right.\\
&&\phantom{\frac{1}{2}}\left.+\left(\left(\bar{x}+\epsilon+1\right)^{\gamma+1}+\left(\bar{y}-\epsilon+1\right)^{\gamma+1}-2\left(\gamma+1\right)t\right)^{-\frac{1}{\gamma+1}}\right)\nonumber\\
&=&\bar{z}^{-1}+2\bar{z}^{-(\gamma+2)}t-\frac{\gamma}{2}\left(\left(\bar{x}+1\right)^{\gamma-1}+\left(\bar{y}+1\right)^{\gamma-1}\right)\bar{z}^{-(\gamma+2)}\epsilon^2\\
&&\phantom{\bar{z}^{-1}}+\left(1+t\right)O(\bar{z}^{-(2\gamma+3)}\epsilon^2).\nonumber
\end{eqnarray}
Because
\begin{equation}
\int_0^\epsilon P\left(\sigma\in dt\right)t \leq \mathbb{E}\sigma=\epsilon^2,
\end{equation}
we can rewrite (\ref{eq:toShowMod2}) as
\begin{eqnarray}
&&\bar{z}^{-1}+2\bar{z}^{-(\gamma+2)}\epsilon^2+O(\bar{z}^{-(2\gamma+3)}\epsilon^2)\\
&\leq&\bar{z}^{-1}+\frac{\gamma}{2}\left(\left(\bar{x}+1\right)^{\gamma-1}+\left(\bar{y}+1\right)^{\gamma-1}\right)\bar{z}^{-(\gamma+2)}\epsilon^2\nonumber
\end{eqnarray}
which holds if
\begin{equation}\label{eq:toShowMod3}
\gamma\left(\left(\bar{x}+1\right)^{\gamma-1}+\left(\bar{y}+1\right)^{\gamma-1}\right)\geq 4.
\end{equation}
Notice that equation (\ref{eq:toShowMod3}) is fulfiled for $\bar{z}$ large enough and $\gamma > 1$, only.\\
\ \\
It remains to show the special case if $\bar{x}$ or $\bar{y}$ is zero. Because of symmetry it is sufficient to treat one of these cases. We assume $\bar{x}=0$ and thus $\bar{y}=(\bar{z}^{\gamma+1}-1)^{1/(\gamma+1)}-1$. Then condition $(\ref{eq:toShowMod1})$ becomes
\begin{eqnarray}
\bar{z}^{-1}&\geq&\phantom{+}\frac{1}{2}\int_0^\infty P\left(\bar{\sigma}\in dt\right)\left|\left|(h(t,\epsilon)+1,h(t,\bar{y}-\epsilon)+1)\right|\right|_{\gamma+1}^{-1}\\
&&+\frac{1}{2}\int_0^\infty P\left(\bar{\sigma}\in dt\right)\left|\left|(1,h(t,\bar{y}+\epsilon)+1)\right|\right|_{\gamma+1}^{-1}.\nonumber
\end{eqnarray}
Applying Taylor approximation in the same way as above results in
\begin{equation}
\bar{z}^{-1}\geq \bar{z}^{-1}-\bar{z}^{-(\gamma+2)}\epsilon+O(\epsilon^2)
\end{equation}
which is true for all $\gamma$ and arbitrary $\bar{z}$.\\
\ \\
The idea, how to transfer the transient result to $Y-X$, is basically equal to the recurrent case on page \pageref{eq:modOdeConst}. This time we consider the process $\hat{X}^\epsilon$, constructed like $X^\epsilon$ but with the modified {\it ode}
\begin{equation}
\frac{d}{dt}f(t)=\left((1-2\epsilon)+b-f(t)\right)^{-\gamma};\ f(0)=0
\end{equation}
instead of (\ref{eq:odeConst}). Equivalently, we define $\hat{Y}^\epsilon(B^\epsilon) := -\hat{X}^\epsilon(-B^\epsilon)$. Again, the proof is not essentially changed by this modifications, and thus, $\hat{Y}^\epsilon-\hat{X}^\epsilon$ is also transient for $\gamma > 1$. Observe that the {\it auxiliary} processes are sandwiched by the original processes:
\begin{equation}
X^{\epsilon'}_t\leq\hat{X}^\epsilon_t+\epsilon\textrm{ and }Y^{\epsilon'}_t\geq\hat{Y}^\epsilon_t-\epsilon
\end{equation}
for all $\epsilon' < \epsilon$. This follows from the same idea as in the recurrent case. It follows
\begin{eqnarray}
Y_t-X_t&=&\lim_{\epsilon\to 0}(Y^\epsilon_t-X^\epsilon_t)\\
&\geq&\hat{Y}^\epsilon_t-\hat{X}^\epsilon_t-2\epsilon,
\end{eqnarray}
which implies the desired result.

\section{Conclusions}\label{sec:conclusions}
In this last section we describe, what our results mean for the {\it opinion game} \cite{bovier06}. We will begin with a short description of the model. Although it is introduced in great generality in the original article, we will adhere to this implementation, which has produced interesting results in the simulations. For a deeper discussion about the choice of the parameters we refer to the original paper. In the second subsection we will point out the connections between our work and the {\it opinion game}.

\subsection{The opinion game}
Bovier et al. consider a generalised, resp. virtual, order book containing the opinion of each participating agent about the value of the share. Here the notion of {\it value} is distinguished from the one of {\it price}. While the price will be determined by the market and is the same for all agents, the value is driven by fundamental and speculative considerations, and thus, varies individually. This is a fundamental difference to the modelling of a classical order book. While a classical order book only keeps track of placed orders, the generalised order book {\it knows} the opinion of all market participants, independent on whether they have made them public. The dynamics of the model are driven by the change of agents' opinion.

A market with $N$ traders trading $M<N$ stocks is considered. For simplification every trader can own at most one share, and furthermore, a discrete time and space setting is assumed. The state of trader $i$ is given by his opinion, denoted by $p_i \in \mathbb{Z}$, and the number of stocks he posseses, $n_i\in\{0,1\}$. A trader with one share is called a {\it buyer}, one without a share is called a {\it seller}. The state of the order book is given by the states of all traders. A state is said to be {\it stable}, if the traders with the $M$ highest opinions posses a share. In particular, one can fully describe the stable state of the order book by the price opinions ${\bf p}:=(p_1,\dots,p_N)$ only. For stable states one can define an ask price as the minimum opinion of all traders possesing a share:
\begin{equation}
p^a := \min\{p_i:n_i = 1\},
\end{equation}    
and the bid price as the maximum opinion of all traders without a share:
\begin{equation}
p^b := \max\{p_i:n_i=0\}.
\end{equation}
The {\it current (logarithmic) price} of the stock is defined by $p:=(p^a-p^b)/2$. The update of the order book state ${\bf p}$ happens in three steps:
\begin{enumerate}
\item At time $(t+1)\in\mathbb{N}_0$, select trader $i$ with probability $g(\cdot;{\bf p}(t),t)$.
\item The selected trader $i$ changes his opinion to $p_i(t)+d$, whereby $d\in\mathbb{Z}$ has distribution $f(\cdot;{\bf p}(t),i,t)$.
\item If ${\bf p}' = (p_1(t),\dots,p_i(t)+d,\cdots,p_N(t))$ is stable, then ${\bf p}(t+1)={\bf p}'$. Otherwise, trader $i$ exchanges his ownership state $n_i(t)$ with the lowest asker, resp. highest bidder $j$. Afterwards, to avoid a direct re-trade, both participants change their opinion away from the trading price.
\end{enumerate}

The function $g$ is defined by
\begin{equation}
g(i;{\bf p}(t),t):=h(p_i(t)-p(t))/Z_g({\bf p}(t)),
\end{equation}
whereby
\begin{equation}\label{eq:distFunc}
h(x):=1/\left(1+|x|\right)^\gamma,\ \gamma>0,
\end{equation}
and $Z_g$ normalizes $g$, s.th. $\sum_{i=1}^N g(i;{\bf p}(t),t) = 1$.

The size of $d$ is chosen from the set $\{-l, \dots, l\}$ with probability
\begin{equation}
f(d;{\bf p}(t),i,t):=\frac{1}{2l+1}\left(\left(\delta_{p_i,p(t)}\delta_{\textrm{ext}}(t)\right)^d \wedge 1\right)\textrm{ for }d\neq 0
\end{equation}
and $f(0;{\bf p}(t),i,t) = 1-\sum_{0<|k|\leq l}f(k;{\bf p}(t),i,t)$. The parameter $\delta_{p_i,p(t)}$ describes the tendency to change the opinion into direction of the price. Thus it is larger than $1$ for $p_i < p$ and smaller for $p_i > p$. The second parameter, $\delta_{\textrm{ext}}$, simulates outer influences on the opinion change, e.g. news or rumors. This force is the same for all traders, but changes its strength in time. Good results were achieved by taking $l = 4$, $\delta_{p_i,p(t)} = \exp(0.1)$ for buyers, and $\delta_{p_i,p(t)} = \exp(-0.1)$ for sellers. The external influence changes its strength after independent, exponentially distributed times with rate $1/2000$ to $\exp(\epsilon_i s'_i)$, whereby $\epsilon_i$ are Bernoulli with $P(\epsilon_i=\pm 1)=1/2$ and $s'_i$ are Exponential with mean $0.12$. Observe that in expectation the external force is slightly stronger than the drift to the price.

The jump away from the trading price in the last step is implemented by setting
\begin{equation}
p_i(t+1) = p^b(t)-k,\ \ p_j(t+1)=p^b(t)+k
\end{equation}
if trader $i$ sells a stock in this step, and
\begin{equation}
p_i(t+1) = p^a(t)+k,\ \ p_j(t+1)=p^a(t)-k
\end{equation}
if he buys it. In the simulations $k$ is a uniformly distributed variable on $\{5,\dots,20\}$.

In the simulations the price is recorded every $100$ opinion updates. Thus, if we talk about one simulation step in the next section, we mean $100$ steps of the underlying dynamics. 

\subsection{Our result in context}
Simulations show that the price process produced by these dynamics has some interesting properties. At first, the distribution of returns, that is the relative change of the price in one step, has heavy tails. Furthermore, the volatility, that is the average size of returns in some time interval, shows correlations on much larger time scales than the implementation would suggest. For the volatility of an interval of size $100$, correlations after $10^4$ steps can be observed. This is suprising, because $10^4$ recorded steps are equal to $10^6$ steps of the dynamics. But the model is Markovian and even the strength of the external influence changes after only $2\cdot 10^3$ steps.

The explanation for these observations can be found in two features of the implementation. As alreday suggested, the external force brings excitement into the market. Else the traders would basically perform random walks into the direction of the price. The returns would be much smaller, an interesting structure of the volatility would not exist. This coincides with the {\it Efficient Market Hypothesis}, because in a world without news and rumors there are no reasons for price changes.

But the external force on its own does not explain the memory of the system in terms of volatility. This behaviour arises from the slower update speed of traders far away from the current price. This mechanism makes sure that the system {\it remembers} price changes on large time scales. If we observe an order book state in which a group of traders has a large distance to the current price, we can conclude, the price must have been in the region of the traders before, as it is very unlikely that a whole group of traders has moved against its drift. Furthermore, after fast price movements the distance between ask and bid price, called {\it gap}, is larger than average and needs some time to recover. In these periods the market is {\it illiquid} and a small number of trades can move the price a lot, which results in an increased volatility. Increased volatility after large price movements is a well observed feature of real world markets.

\begin{figure}[htb]
\begin{center}
\frame{\includegraphics[width=.45\textwidth]{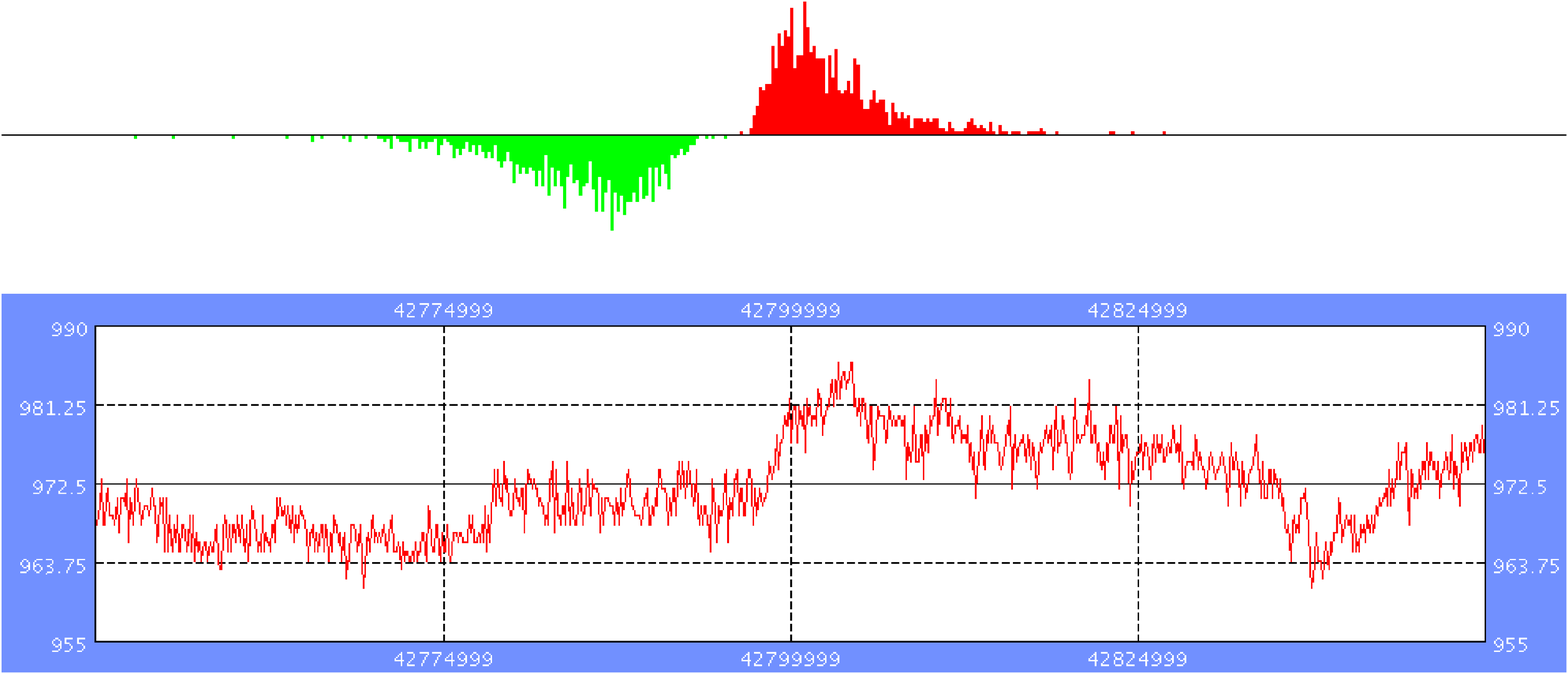}}
\frame{\includegraphics[width=.45\textwidth]{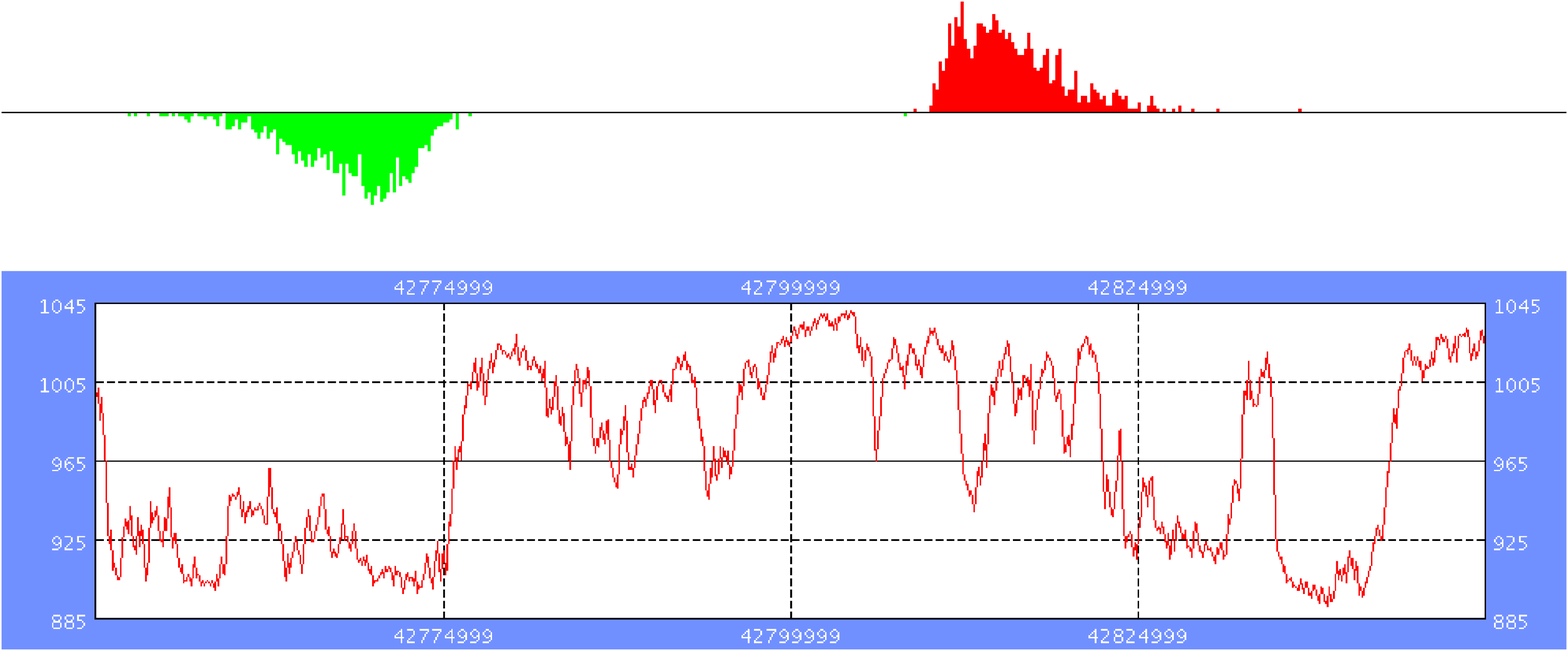}}
\caption{Screenshots of the virtual order books after $428500$ simulation steps for $\gamma = 1.5$ (left) and $\gamma = 1.6$ (right) with same initial conditions and same realisation of external influences. Observe the different distances between buyers (green) and sellers (red) and the different behaviour of the price processes (blue box).}
\label{fig:conf}
\end{center}
\end{figure}

\begin{figure}[htb]
\begin{center}
\includegraphics[width=.45\textwidth]{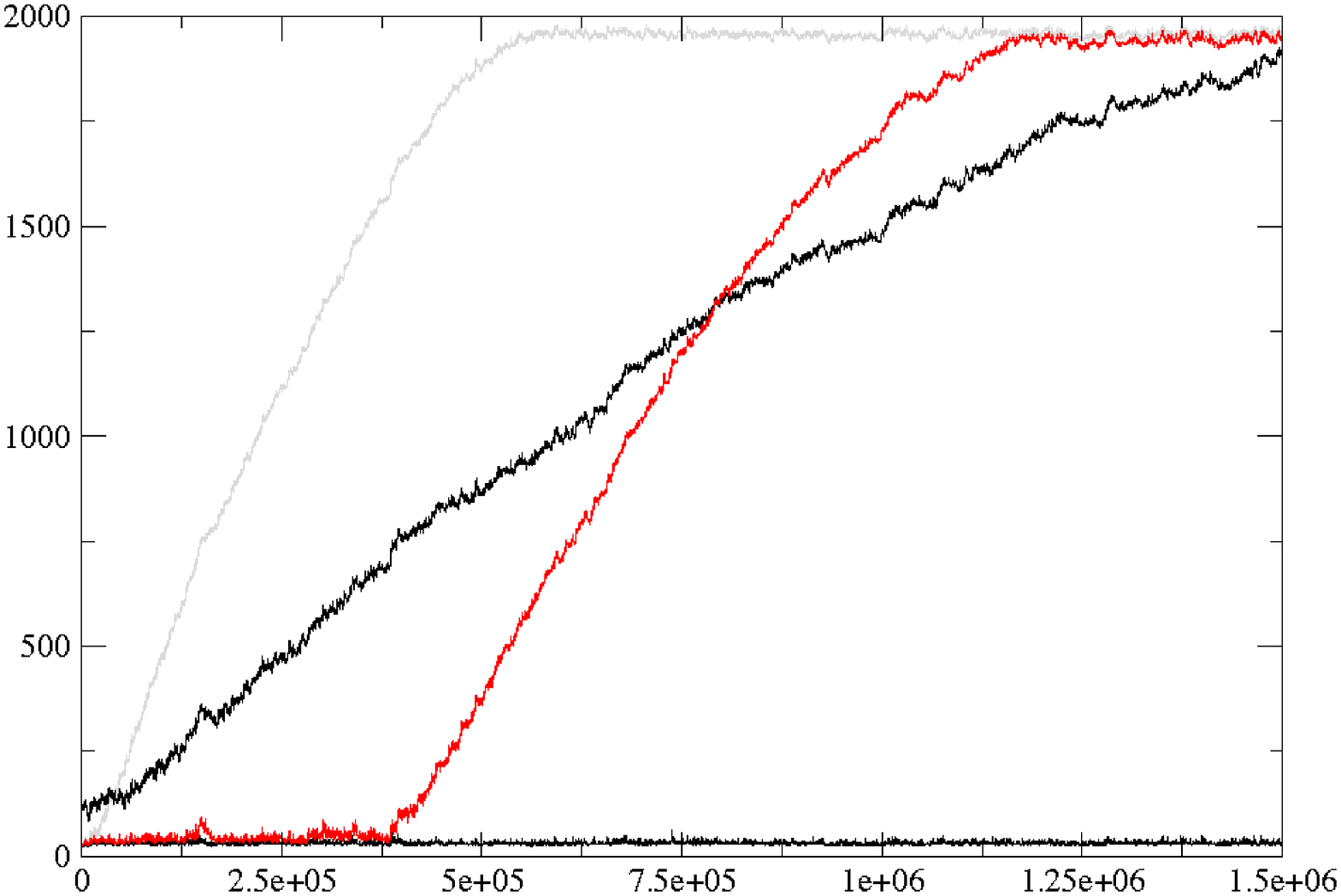}
\includegraphics[width=.45\textwidth]{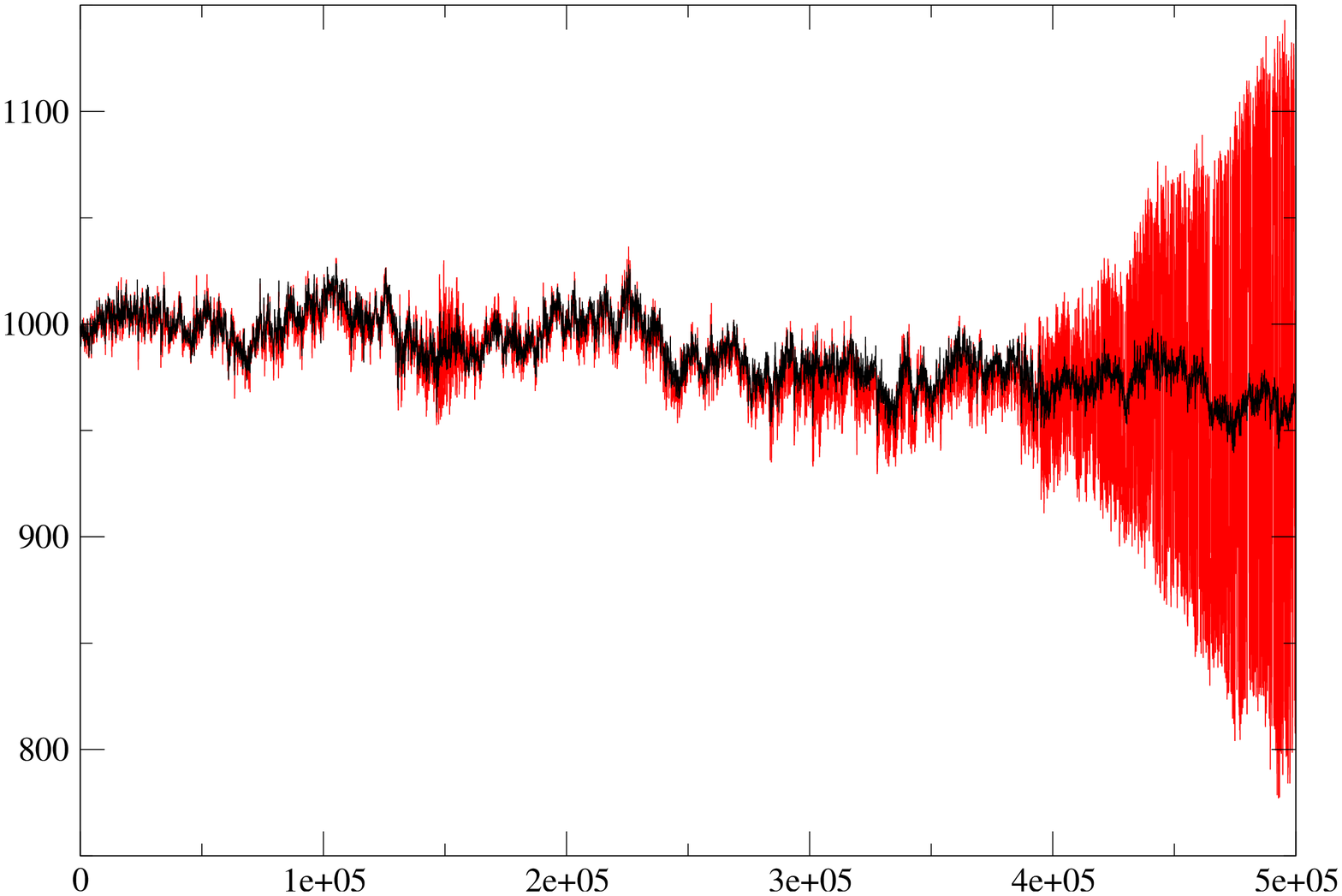}\\
\caption{The left graph shows the {\it gap} of the system for different $\gamma$. While it is stable for $\gamma = 1.5$ (black lower graph), it increases for $\gamma = 1.6$ (red) and $\gamma = 1.7$ (grey). However, if the system is started with $\gamma = 1.5$ but with an artificially enlarged gap, it also increases (black increasing graph). The convergence to a value below $2000$ is due to a restriction of the state space in the numerical simulations. The right graph shows the stable resp. unstable behaviour for $\gamma = 1.5$ (black) and $\gamma = 1.6$ (red) in terms of the price process.}\label{fig:price}
\end{center}
\end{figure}

Thus the connection of update speed and distance to the price is of paramount importance for the model. Indeed, the larger $\gamma$ is chosen in formula (\ref{eq:distFunc}) the better the just explained phenomena can be observed. However, a larger $\gamma$ contains the risk of instablity of the whole system. It turns out that once the gap has exceeded a certain size (depending on $\gamma$), it cannot recover anymore and the two groups, buyers and sellers, drift away from each other. Then the price waves between these groups, driven by two traders, one from each group, which were able to get away and now basically move according to the external drift without any resistance by surrounding traders. For $\gamma \geq 1.6$ this happens quite fast while the model has remained stable in simulations over several days for $\gamma = 1.5$ (Figure \ref{fig:conf}). On the other hand, if we start a simulation already with a large gap and $\gamma = 1.5$, also this system is not able to recover. As a large gap size will eventually reached by randomness, it is justified to talk about a {\it metastable} behaviour.  In Figure \ref{fig:price} we illustrate these statements with a sample. Instead of recording the difference between ask and bid price, we have taken the distance between the 950th and the 1050th trader ordered by their opinions (i.e. the buyer with the 50th highest opinion and the seller with the 50th lowest one), because traders close to the price suffer much more fluctuations than agents with some distance. In this sense our choice represents the majority of the traders. 

In the situation when the trader groups have already a large distance from each other, the two traders in between, and also the price, perform basically a random walk. Especially, when the two traders are close to the middle in between both groups, their probability to move is almost $1$. In this case our model with a Brownian motion as driving force offers a reasonable approximation for the behaviour of the system. Thus, our results give few hope that any simulation with $\gamma > 1$ will be stable forever. But for $\gamma < 1$ the memory effect producing all the statistical facts is too small. However, as already mentioned, the model seems to be stable on a large time scale for $\gamma = 1.5$. This and also the sharp threshold between $1.5$ and $1.6$ are not understood. More research is neccessary here.

Besides these findings the three particle model introduced in this paper has its qualities on its own. As a simple model for longterm investors, this easy setting already exhibits an interesting and non-trivial longterm behaviour. As a logical next step it will be interesting to see, how the results change if we substitute the Brownian motion by a L\'evy process, which is much more realistic for price process on stock markets.

\appendix
\section{Proof of Lemma \ref{lem:limExists}}\label{sec:appendix}
We turn to the proof of Lemma \ref{lem:limExists}:
\begin{quote}
Let $S\subset [0,\infty)$ be a compact set and $\epsilon \ll \exp(-\gamma\cdot \sup S)$. Then  
\begin{equation}
\sup_{t\in S}\left|X^{\epsilon'}_t-X^\epsilon_t\right| \leq \epsilon K_S\textrm{ a.s.,}
\end{equation}
whereby $K_S$ is a finite, deterministic constant depending on $S$, and $\epsilon' < \epsilon$.
\end{quote}
Because $S$ is compact, {\it w.l.o.g.}, we may assume $S = [0,t^*]$ for some $0\leq t^*<\infty$. Remember that the jump times of $B^\epsilon$ were denoted by $\bar{\sigma}^\epsilon$ in (\ref{eq:sigmaBar}), and the time between two jumps by $\sigma^\epsilon$ in (\ref{eq:sigma}). Furthermore,
\begin{equation}\label{eq:distX2B}
\left|B^\epsilon-B^{\epsilon'}\right|<\epsilon.
\end{equation}
We denote the distance of $X^\epsilon$ to $B^\epsilon$ by
\begin{equation}
d_i := B^\epsilon_{\bar{\sigma}_i}-X^\epsilon_{\bar{\sigma}_i},
\end{equation}
and the distance to $X^{\epsilon'}$ by
\begin{equation}
\Delta_i := X^\epsilon_{\bar{\sigma}_{i}} - X^{\epsilon'}_{\bar{\sigma}_{i}},
\end{equation}
always meaning $\bar{\sigma}$ with respect to $\epsilon$. We would like to maximize $\Delta_2$, thus, we assume that $B^\epsilon$ has jumped upwards at $\bar{\sigma}_1$. Then $d_1 = \epsilon$ and $|\Delta_1| < \epsilon$. We first assume that $\Delta_1$ is positive. By definition of $\Delta$ and of $\bar{h}$ in (\ref{eq:hFuncBar}),
\begin{eqnarray}
\Delta_{2}&=&\left(X^\epsilon_{\bar{\sigma}_{2}}-X^\epsilon_{\bar{\sigma}_1}\right)-\left(X^{\epsilon'}_{\bar{\sigma}_{2}}-X^{\epsilon'}_{\bar{\sigma}_1}\right)+\left(X^\epsilon_{\bar{\sigma}_1}-X^{\epsilon'}_{\bar{\sigma}_1}\right)\\
&\stackrel{(\ref{eq:distX2B})}{\leq}&\bar{h}(\sigma_1,d_1)-\bar{h}(\sigma_1,d_1+\Delta_1+\epsilon)+\Delta_1\\
&\stackrel{(\ref{eq:hFunc})}{=}&h(\sigma_1,d_1+\Delta_1+\epsilon)-h(\sigma_1,d_1)-\epsilon\label{eq:hFormula}.
\end{eqnarray}
Remember that $h$ is {\it basically} defined as
\begin{equation}
h(t,d)=\left(\left(d+1\right)^{\gamma+1}-\left(\gamma+1\right)t\right)^{1/(\gamma+1)}-1.
\end{equation} 
As the distance will not increase anymore, once $X^\epsilon$ has hit $B^\epsilon$, we get an upper bound for $\sigma_1$:
\begin{equation}
h(\sigma_1,d_1)\geq 0\ \Leftrightarrow\ \sigma_1 \leq \frac{(d_1+1)^{\gamma+1}-1}{\gamma +1}.
\end{equation}
Because $d_1 = \epsilon$ we have $\sigma_1 \leq \epsilon$. As $d_1$, $\Delta_1$, $\epsilon$, and $\sigma_1$ are small in comparison to $1$, we apply Taylor twice to line (\ref{eq:hFormula}) and get
\begin{eqnarray}
\Delta_2&\leq&\left(1-\left(\gamma+1\right)\sigma_1\right)^{-\gamma/(\gamma+1)}\left(\Delta_1+\epsilon\right)-\epsilon\\
&=&\Delta_1+\gamma\left(\Delta_1+\epsilon\right)\sigma_1\\
&=&\Delta_1\left(1+\gamma\epsilon\right).
\end{eqnarray}
With the same argumentation we can conclude that
\begin{equation}
\Delta_{i+1}\leq \Delta_i\left(1+\gamma\epsilon\right),
\end{equation}
and thus,
\begin{eqnarray}
X^\epsilon_{t^*}-X^{\epsilon'}_{t^*}&=&\Delta_{t^*/\epsilon}\\
&\leq&\Delta_1\left(1+\gamma\epsilon\right)^{t^*/\epsilon}\\
&\to&\epsilon e^{\gamma t^*}.
\end{eqnarray}
On the other hand, if $X^{\epsilon'}>X^\epsilon$, basically the same idea applies: the distance grows the quickest, if one of the processes always stays close to its attracting process s.th. it has drift speed $1$. Now, if $X^{\epsilon'}$ increases with speed $1$ (as a worst case assumption), $\sigma^\epsilon = \epsilon$ and we end up with the same calculation as before.

It should be mentioned that our estimations are rough, as we do not consider the structure of Brownian paths, but only the worst case of all continous paths. However, uniform convergence on compact intervals is the best one can get and every improvement would only change the constant $K_S$.
\bibliographystyle{amsplain}
\bibliography{weiss-stalker}
\end{document}